\documentclass[letterpaper, 10 pt, conference]{ieeeconf}  

\IEEEoverridecommandlockouts                              
\usepackage{times} 
\usepackage{enumerate}
\usepackage{mathtools}
\usepackage{amsmath} 
\usepackage{amssymb}  
 
\usepackage{amsthm}
\usepackage{graphicx}
\usepackage{xcolor}
\usepackage{url}
\usepackage{algorithm}
\usepackage{algpseudocode}
\usepackage{cite}

\usepackage{abbreviations}

\usepackage{siunitx}
\usepackage{booktabs, makecell, multirow}
\usepackage{aligned-overset}
\usepackage{tikz,pgfplots}
\usetikzlibrary{
    calc, 
    decorations.markings
}
\makeatletter
\tikzset{
  nomorepostactions/.code={\let\tikz@postactions=\pgfutil@empty},
  mymark/.style 2 args={decoration={markings,
    mark= between positions 0 and 1 step (1/11)*\pgfdecoratedpathlength with{%
        \tikzset{#2,every mark}\tikz@options
        \pgfuseplotmark{#1}%
      },  
    },
    postaction={decorate},
    /pgfplots/legend image post style={
        mark=#1,mark options={#2},every path/.append style={nomorepostactions}
    },
  },
}
\makeatother
\pgfplotsset{compat=newest}
\usepackage{nicefrac}

\usepackage[bb=dsserif]{mathalpha}

\newtheorem{theorem}{Theorem}
\newtheorem{corollary}[theorem]{Corollary}
\newtheorem{assumption}[theorem]{Assumption}

\newcommand{\SOS}{\mathrm{SOS}}

\newcommand\copyrighttext{%
  \footnotesize 
  \textcopyright This version has been accepted for publication in the European Journal of Control, 2025. Personal use of this material is permitted. Permission from EUCA must be obtained for all other uses, in any current or future media, including reprinting/republishing this material for advertising or promotional purposes, creating new collective works, for resale or redistribution to servers or lists, or reuse of any copyrighted component of this work in other works.
}
\newcommand\copyrightnotice{%
    \begin{tikzpicture}[remember picture,overlay]
        \node[anchor=south,yshift=10pt] at (current page.south) {\fbox{\parbox{\dimexpr\textwidth-\fboxsep-\fboxrule\relax}{\copyrighttext}}};
    \end{tikzpicture}%
}

\title{\LARGE \bf
Koopman-based control using sum-of-squares optimization:\\Improved stability guarantees and data efficiency
}

\author{Robin Str\"asser, Julian Berberich, Frank Allg\"ower
\thanks{F.\ Allgöwer is thankful that this work was funded by the Deutsche Forschungsgemeinschaft (DFG, German Research Foundation) under Germany's Excellence Strategy -- EXC 2075 -- 390740016 and within grant AL 316/15-1 -- 468094890. 
R.\ Strässer thanks the Graduate Academy of the SC SimTech for its support.}
\thanks{Robin Str\"asser, Julian Berberich, and Frank Allg\"ower are with the Institute for Systems Theory and Automatic Control, University of Stuttgart, 70550 Stuttgart, Germany
		(email:\{robin.straesser, julian.berberich, frank.allgower\}@ist.uni-stuttgart.de).}%
}

\begin{document} 
 
\maketitle
\thispagestyle{empty}
\pagestyle{empty}


\begin{abstract}%
    In this paper, we propose a novel controller design approach for unknown nonlinear systems using the Koopman operator.
    In particular, we use the recently proposed stability- and feedback-oriented extended dynamic mode decomposition (SafEDMD) architecture to generate a data-driven bilinear surrogate model with certified error bounds. 
    Then, by accounting for the obtained error bounds in a controller design based on the bilinear system, one can guarantee closed-loop stability for the true nonlinear system. 
    While existing approaches over-approximate the bilinearity of the surrogate model, thus introducing conservatism and providing only local guarantees, we explicitly account for the bilinearity by using sum-of-squares (SOS) optimization in the controller design.
    More precisely, we parametrize a rational controller stabilizing the error-affected bilinear surrogate model and, consequently, the underlying nonlinear system. 
    The resulting SOS optimization problem provides explicit data-driven controller design conditions for unknown nonlinear systems based on semidefinite programming.
    Our approach significantly reduces conservatism by establishing a larger region of attraction and improved data efficiency. 
    The proposed method is evaluated using numerical examples, demonstrating its advantages over existing approaches.
\end{abstract}

\begin{keywords}
    Data-driven control, Koopman operator, sum-of-squares optimization, nonlinear systems
\end{keywords}

\copyrightnotice\vspace*{-\baselineskip}
%
\section{Introduction}\label{sec:introduction}
Nonlinear control systems pose significant challenges in system analysis and controller design due to their complex behaviors~\cite{khalil:2002}.
Traditional model-based control methods may struggle with these systems due to their inherent nonlinearity, high-dimensional dynamics, and the difficulty in obtaining accurate mathematical models.
Recently, data-driven methods have become more important as they leverage available data to model and control nonlinear systems without relying solely on predefined mathematical models~\cite{martin:schon:allgower:2023b}.
These methods help overcome the limitations of traditional control techniques, allowing for more flexible and reliable control.

One effective data-driven strategy for controlling nonlinear systems is based on the Koopman operator~\cite{koopman:1931}.
The Koopman operator is a mathematical tool that allows for representing the nonlinear dynamics of a system in a higher-dimensional space, where its behavior can be analyzed using linear methods~\cite{mauroy:mezic:susuki:2020}. 
This approach has gained attention in recent years for its ability to analyze and design controllers for nonlinear systems via (bi-)linear systems theory~\cite{budisic:mohr:mezic:2012}.
Considered applications areas include fluid dynamics~\cite{mezic:2013}, robotics~\cite{bruder:remy:vasudevan:2019}, biological systems~\cite{golany:radinstky:freedman:minha:2021}, and cryptography~\cite{strasser:schlor:allgower:2023}, where traditional methods are often infeasible due to the high complexity of the systems.
Since the Koopman operator is inherently infinite dimensional, a common approach is to approximate its action on a finite-dimensional space of observable functions using techniques like extended dynamic mode decomposition (EDMD)~\cite{williams:kevrekidis:rowley:2015,proctor:brunton:kutz:2018}.

In practice, the approximation error of the EDMD-based surrogate model is often neglected. 
Despite this simplification, successful results have been achieved in many applications~\cite{bevanda:sosnowski:hirche:2021}. 
However, without rigorous analysis of the approximation error, no theoretical guarantees can be given on stability, performance, and robustness. 
Hence, it is important to account for the approximation error explicitly to ensure the desired behavior of the designed controller in closed loop with the true nonlinear system.
There exist some first results addressing the approximation error in finite-dimensional representations of the Koopman operator, e.g., via integral quadratic constraints~\cite{eyuboglu:karimi:2024,eyuboglu:strasser:allgower:karimi:2025} or the stability and feedback-oriented EDMD (SafEDMD) framework~\cite{strasser:berberich:allgower:2023a,strasser:schaller:worthmann:berberich:allgower:2025,strasser:schaller:worthmann:berberich:allgower:2024b}.
The latter provides a data-driven bilinear surrogate model with guaranteed error bounds that can be used in both continuous and discrete time, and there are preliminary extensions to the presence of stochastic noise~\cite{chatzikiriakos:strasser:iannelli:allgower:2024} and the combination with model predictive control~\cite{worthmann:strasser:schaller:berberich:allgower:2024}. 
While these approaches provide a starting point for rigorous Koopman-based control, the controller design relies on over-estimating the bilinearity of the surrogate model and the approximation error to obtain an uncertain linear representation.
Hence, the resulting controller tends to be conservative, i.e., the resulting guaranteed region of attraction (RoA) can be small and the approaches suffer from poor robustness.

In this paper, we propose a novel approach to design controllers for unknown nonlinear systems based on the Koopman operator and sum-of-squares optimization (SOS).
Here, we first apply SafEDMD~\cite{strasser:schaller:worthmann:berberich:allgower:2024b} to derive a data-driven bilinear surrogate model with guaranteed error certificates.
Then, we design a rational controller globally stabilizing the error-affected bilinear surrogate model and, thus, the true underlying system.
To this end, we formulate a novel SOS program providing a convex controller design for unknown nonlinear systems based on data samples.
Compared to existing Koopman-based control methods with stability guarantees, which over-approximate the bilinearity of the surrogate model, our approach explicitly accounts for the bilinearity.
Thus, our method significantly reduces the conservatism over these approaches, leading to a larger RoA and a more efficient usage of data, i.e., less data is sufficient for a feasible controller design.
Further, we showcase the benefits of the proposed approach compared to existing methods in numerical examples.
We note that the existing methods in~\cite{nobuyama:aoyagi:kami:2011,vatani:hovd:olaru:2014} also used SOS optimization to control bilinear systems. 
However, these approaches only address nominal bilinear systems, whereas we propose a robust controller design which can cope with the uncertainty due to the Koopman operator approximation.

The paper is organized as follows. 
In Section~\ref{sec:problem-setting}, we introduce the considered class of nonlinear systems and the corresponding data-driven system representation based on SafEDMD.
Then, in Section~\ref{sec:controller-design-bilinear}, we first present the proposed SOS-based controller design for the special case of bilinear systems with known system dynamics to ease the exposition.
The general case of data-driven control for unknown nonlinear systems is addressed in Section~\ref{sec:controller-design-nonlinear}.
Finally, we illustrate the proposed approach in numerical examples in Section~\ref{sec:numerical-example}, before concluding the paper in Section~\ref{sec:conclusion}.

\emph{Notation:}
We write $I_p$ for the $p\times p$ identity matrix and $0_{p\times q}$ for the $p\times q$ zero matrix, where we omit the index if the dimension is clear from the context. 
A vector of ones is denoted by $\mathbb{1}_d\in\bbR^d$. 
For symmetric $A$, we write $A \succ 0$ or $A \succeq 0$ if $A$ is positive definite or positive semidefinite, respectively. 
Negative (semi)definiteness is defined analogously. 
Matrix blocks which can be inferred from symmetry are denoted by $\star$ and we abbreviate $B^\top AB$ by writing $[\star]^\top AB$.
The Kronecker product is denoted by $\kron$. 
Further, we denote by $\bbR[x,d]$ the set of all polynomials $s$ in the variable $x\in\bbR^n$ with degree $d$ and real coefficients, i.e., $s(x) = \sum_{\alpha\in\bbN_0^n, |\alpha|\leq d} s_\alpha x^\alpha$ with $s_\alpha\in\bbR$ and $|\alpha|\leq d\in\bbN_0$.
Here, $x^\alpha = x_1^{\alpha_1}\cdots x_n^{\alpha_n}$ are monomials for a vectorial index $\alpha\in\bbN_0^n$, where $|\alpha|= \alpha_1 + \cdots + \alpha_n$.
The set of all $p\times q$-matrices with elements in $\bbR[x,d]$ is denoted by $\bbR[x,d]^{p\times q}$.
If a matrix $S\in\bbR[x,2d]^{p\times p}$ can be decomposed as $S = T^\top T$ for some $T\in\bbR[x,d]^{q\times p}$, then $S$ is called an SOS matrix in $x$ and we write $S\in\SOS[x,2d]^p$.
Finally, we say $S\in\bbR[x,2d]^{p\times p}$ is strictly SOS if $(S-\varepsilon I_p)\in\SOS[x,2d]^p$ for some $\varepsilon>0$, denoted by $S\in\SOS_+[x,2d]^p$. 
%
\section{Problem setting}\label{sec:problem-setting} 
We consider \emph{unknown} nonlinear control-affine systems of the form
\begin{equation}\label{eq:dynamics-nonlinear}
    \dot{x}(t) 
    = f_c(x(t),u(t)) 
    = f(x(t)) + \sum_{i=1}^mg_i(x(t)) u_i(t)
\end{equation}
with state $x\in\bbR^n$, input $u\in\bbR^m$, initial condition $x(0)\in\bbR^n$, and unknown functions $f$ and $g_i$, $i \in 1,...,m$.
Further, we assume that $f(0)=0$ holds, i.e., that the origin is a controlled equilibrium for $u=0$.
\begin{assumption}\label{ass:continuity-dynamics}
    The vector field $f_c:\bbR^n\times\bbR^m\to\bbR^n$ is continuous and locally Lipschitz in its first argument around the origin.
\end{assumption}
Assumption~\ref{ass:continuity-dynamics} is satisfied for many practical systems and ensures a smooth behavior between sampled data points.
The goal of this paper is to design a state-feedback controller $u=\mu(x)$ for the unknown nonlinear system~\eqref{eq:dynamics-nonlinear} that is guaranteed to stabilize the origin of the system.
To this end, we first find a suitable data-driven system representation, for which we leverage Koopman operator theory and its finite-dimensional approximation. 
More precisely, we rely on SafEDMD~\cite{strasser:schaller:worthmann:berberich:allgower:2024b} which uses a lifted bilinear representation of system~\eqref{eq:dynamics-nonlinear} with the lifted state $\Phi(x) = \begin{bmatrix} x^\top & \phi_{n+1} & \cdots & \phi_N\end{bmatrix}^\top$ and the corresponding (discrete-time) dynamics
\begin{equation}\label{eq:dynamics-nonlinear-bilinear}%
    \Phi(x_{k+1}) = A\Phi(x_k) + B_0 u_k + \tB (u_k\kron \Phi(x_k)) + r(x_k,u_k)
\end{equation}
with $x_{k} = x(k\Delta t)$ for some sampling rate $\Delta t>0$ and $k\in\bbN_0$.
Here, the system matrices $A,B_0,\tB$ are estimated by data samples from repeated experiments with different constant inputs $u(t)\equiv\bar{u}$, i.e., we consider the dataset 
\begin{equation}\label{eq:dataset}
    \cD=\left\{
        \{x^{\bar{u}}(t_j),x^{\bar{u}}(t_j+\Delta t)\}_{j=1}^{d} 
        \;\text{for}\; 
        \bar{u}\in\{0,e_1,...,e_m\}
    \right\}
\end{equation} 
with unit vectors $e_i\in\bbR^m$. 
Based on the data $\cD$, we define 
\begin{align*}
    \Phi_X^{\bar{u}} &= \begin{bmatrix}\Phi(x^{\bar{u}}(t_1)) & \cdots & \Phi(x^{\bar{u}}(t_d))\end{bmatrix},\\
    \Phi_{X_+}^{\bar{u}} &= \begin{bmatrix}\Phi(x^{\bar{u}}(t_1+\Delta t)) & \cdots & \Phi(x^{\bar{u}}(t_d+\Delta t))\end{bmatrix},
\end{align*}
and solve the linear regression problems 
\begin{subequations}\label{eq:EDMD-linear-regression}
    \begin{align}
        A 
        &= \argmin_{A\in\bbR^{N\times N}} 
        \| \Phi_{X_+}^{0} - A \Phi_X^{0} \|_\mathrm{F},
        \\
        (B_{0,i},B_i) 
        &= \argmin_{\substack{B_{0,i}\in\bbR^N,\\B_i\in\bbR^{N\times N}}} 
        \left\| \Phi_{X_+}^{e_i} 
            - \begin{bmatrix} B_{0,i} & B_i \end{bmatrix} 
            \begin{bmatrix} \mathbb{1}_d^\top \\ \Phi_X^{e_i} \end{bmatrix}
        \right\|_\mathrm{F}.
    \end{align}
\end{subequations}
Then, we stack the matrices $B_{0,i}$, $B_i$ and define $B_0=\begin{bmatrix}B_{0,1} & \cdots & B_{0,m}\end{bmatrix}$ and $\tB = \begin{bmatrix} B_1 - A & \cdots & B_m - A \end{bmatrix}$.
Further, the residual can be proven to satisfy the proportional bound
\begin{equation}\label{eq:proportional-bound-nonlinear}
    \|r(x,u)\| \leq c_x \|\Phi(x)\| + c_u \|u\|
\end{equation}
for some $c_x,c_u>0$, see~\cite[Cor.~3.2]{strasser:schaller:worthmann:berberich:allgower:2024b} for details.
%
\section{SOS-based controller design for bilinear systems}\label{sec:controller-design-bilinear}
Next, we propose a rational state-feedback controller based on SOS optimization. 
To this end, we first restrict ourselves to discrete-time bilinear systems, i.e., system~\eqref{eq:dynamics-nonlinear-bilinear} with $\Phi(x)=x$ leading to 
\begin{equation}\label{eq:dynamics-bilinear}
    x_+ = A x + B_0 u + \tB (u\kron x) + r(x,u),
\end{equation}
where for brevity we write $(x_+,x,u)$ for $(x_{k+1},x_k,u_k)$. 
In this section, we assume that the system matrices $A,B_0,\tB$ are known and that the proportional bound
\begin{equation}\label{eq:proportional-bound-bilinear}
    \|r(x,u)\| \leq c_x \|x\| + c_u\|u\|
\end{equation}
on the unknown residual $r(x,u)$ holds true.
As shown in~\cite[Prop.~3, Prop.~4]{vatani:hovd:olaru:2014}, an unstable discrete-time bilinear system can only be (globally) stabilized by a rational control law if the degrees of the polynomial in the numerator and the polynomial in the denominator are the same. 
Hence, neither a linear nor a polynomial state feedback is sufficient to stabilize the system.
The following theorem establishes a globally exponentially stabilizing controller for the uncertain bilinear system~\eqref{eq:dynamics-bilinear} based on SOS optimization.
\begin{theorem}\label{thm:controller-design-bilinear}
    Suppose constants $c_x,c_u>0$ such that~\eqref{eq:proportional-bound-bilinear} holds.
    If there exist $\alpha>0$, $P=P^\top\succ 0$ of size $n\times n$, $L_\mathrm{n}\in\bbR[x,2\alpha-1]^{m\times n}$, $\tau\in\SOS_+[x,2\alpha]$, $u_\dd\in\SOS_+[x,2\alpha]$, and $\rho>0$ such that~\eqref{eq:controller-design-bilinear}
    \begin{figure*}[tb]
        \begin{equation}\label{eq:controller-design-bilinear}
            \begin{bmatrix}
                u_\dd(x)P - \tau(x)I_n & 0 & 0 & u_\dd(x)AP + B_0 L_\mathrm{n}(x) + \tB (L_\mathrm{n}(x)\kron x) \\
                \star & \frac{\tau(x)}{2c_x^{2}}I_n & 0 & u_\dd(x) P \\
                \star & \star & \frac{\tau(x)}{2c_u^{2}} I_m & L_\mathrm{n}(x) \\
                \star & \star & \star & u_\dd(x) (P - \rho I_n)
            \end{bmatrix}
            \in\SOS[x,2\alpha]^{3n+m}
        \end{equation}
        \medskip
        \hrule
    \end{figure*}
    holds,
    then the controller
    \begin{equation}\label{eq:controller-bilinear-explicit}
        \mu(x) = \frac{1}{u_\dd(x)}L_\mathrm{n}(x)P^{-1} x
    \end{equation}
    exponentially stabilizes the origin of the bilinear system~\eqref{eq:dynamics-bilinear}.
\end{theorem}

\begin{proof}
    We consider the Lyapunov function $V(x)=x^\top P^{-1} x$ with $P^{-1}\succ 0$ and want to show $\Delta V(x) = x_+^\top P^{-1} x_+ - x^\top P^{-1} x \leq -\varepsilon\|x\|^2$ for all $x\in\bbR^n$ with some $\varepsilon>0$.

    To this end, we define $K_\mathrm{n}(x)=L_\mathrm{n}(x)P^{-1}$ and observe  
    \begin{equation*}
        L_\mathrm{n}(x) \kron x 
        = K_\mathrm{n}(x) P \kron x 
        = (K_\mathrm{n}(x) \kron x) P
    \end{equation*}
    using $P=P\kron 1$.
    Then, defining
    \begin{equation*}
        \Xi(x) = u_\dd(x)A + B_0 K_\mathrm{n}(x) + \tB (K_\mathrm{n}(x)\kron x),
    \end{equation*}
    substituting $K_\mathrm{n}(x)$ and $\Xi(x)$ into~\eqref{eq:controller-design-bilinear}, and applying the Schur complement (cf.~\cite{boyd:vandenberghe:2004}) yields
    \footnotesize
    \begin{equation*}
        \hspace*{-0.01\linewidth}
        \begin{bmatrix}
            u_\dd(x) P - \tau(x) I_n & 0 & 0 & 0 & \Xi(x)P \\
            \star & \frac{\tau(x)}{2c_x^{2}}I_n & 0 & 0 & u_\dd(x) P \\
            \star & \star & \frac{\tau(x)}{2c_u^{2}} I_m & 0 & K_\mathrm{n}(x)P \\
            \star & \star & \star & \frac{u_\dd(x)}{\rho} P^2 & u_\dd(x) P \\
            \star & \star & \star & \star & u_\dd(x) P
        \end{bmatrix}
        \succeq 0
    \end{equation*}
    \normalsize
    for all $x\in\bbR^n$.
    Next, we apply the Schur complement to obtain
    \begin{multline}\label{eq:stability-condition-SOS-Schur}
        \diag\left(
            u_\dd(x) P - \tau(x) I_n, 
            \frac{\tau(x)}{2c_x^{2}}I_n, 
            \frac{\tau(x)}{2c_u^{2}} I_m, 
            \frac{u_\dd(x)}{\rho} P^2
        \right)
        \\
        - \begin{bmatrix}
            \Xi(x)P \\
            u_\dd(x) P \\
            K_\mathrm{n}(x)P \\
            u_\dd(x) P
        \end{bmatrix}
        (u_\dd(x) P)^{-1}
        \left[\star\right]^\top
        \succeq 0.
    \end{multline}
    We decompose~\eqref{eq:stability-condition-SOS-Schur} into three terms, which we write as
    \begin{multline*}
        \begin{bmatrix}
            u_\dd(x) P & 0 & 0 & 0 \\
            0 & 0 & 0 & 0 \\
            0 & 0 & 0 & 0 \\
            0 & 0 & 0 & 0
        \end{bmatrix}
        - \begin{bmatrix}
            \Xi(x) \\
            u_\dd(x) I_n \\
            K_\mathrm{n}(x) \\
            u_\dd(x) I_n
        \end{bmatrix}
        \frac{1}{u_\dd(x)} P
        \left[\star\right]^\top
        \\
        = 
        \begin{bmatrix}
            \Xi(x) & -I_n \\
            u_\dd(x) I_n & 0 \\
            K_\mathrm{n}(x) & 0 \\
            u_\dd(x) I_n & 0
        \end{bmatrix}
        \frac{1}{u_\dd(x)}\Pi_\mathrm{stab}(x)^{-1}
        \left[\star\right]^\top
    \end{multline*}
    and 
    \begin{multline*}
        \diag\left( 
            -\tau(x) I_n, 
            \frac{\tau(x)}{2c_x^2} I_n, 
            \frac{\tau(x)}{2c_u^2} I_m, 
            0
        \right)
        \\
        = \begin{bmatrix}
            0 & 0 & I_n \\
            -I_n & 0 & 0 \\
            0 & -I_m & 0 \\
            0 & 0 & 0 
        \end{bmatrix}
        \tau(x) \Pi_r^{-1} 
        \left[\star\right]^\top
    \end{multline*}
    and
    \begin{equation*}
        \begin{bmatrix}
            0 & 0 & 0 & 0 \\
            0 & 0 & 0 & 0 \\
            0 & 0 & 0 & 0 \\
            0 & 0 & 0 & \frac{u_\dd(x)}{\rho} P^2
        \end{bmatrix}
        = \begin{bmatrix}
            0 \\ 0 \\ 0 \\ -I_n
        \end{bmatrix}
        \frac{u_\dd(x)}{\rho} P^2
        \left[\star\right]^\top,
    \end{equation*}
    where 
    \begin{align*}
        \Pi_\mathrm{stab}(x) 
        &= \diag\left( -P^{-1}, \frac{1}{u_\dd(x)^2}P^{-1} \right),
        \\
        \Pi_r
        &= \diag\left(2c_x^2 I_n, 2c_u^2 I_m, -I_n\right).
    \end{align*}
    Hence,~\eqref{eq:stability-condition-SOS-Schur} is equivalent to
    \begin{equation*}
    \hspace*{-0.02\linewidth}
        \Psi^\top
        \left[\def\arraystretch{1.2}\begin{array}{c|c|c}
            \frac{1}{u_\dd(x)}\Pi_\mathrm{stab}(x)^{-1} & 0 & 0\\\hline
            0 & \tau(x) \Pi_r^{-1} & 0 \\\hline 
            0 & 0 & \frac{u_\dd(x)}{\rho} P^2
        \end{array}\right]
        \Psi
        \succeq 0
    \end{equation*}
    with
    \begin{equation*}
    \hspace*{-0.02\linewidth}
        \Psi^\top 
        = 
        \left[\def\arraystretch{1.2}\begin{array}{cc|ccc|c}
            \Xi(x) & -I_n & 0 & 0 & I_n & 0 \\
            u_\dd(x) I_n & 0 & -I_n & 0 & 0 & 0 \\
            K_\mathrm{n}(x) & 0 & 0 & -I_m & 0 & 0 \\
            u_\dd(x) I_n & 0 & 0 & 0 & 0 & -I_n
        \end{array}\right].
    \end{equation*}
    We apply the dualization lemma~\cite[Lm.~4.9]{scherer:weiland:2000} and obtain
    \begin{equation}\label{eq:proof-primal-inequality}
        \tilde{\Psi}^\top
        \left[\def\arraystretch{1.2}\begin{array}{c|c|c}
            u_\dd(x)\Pi_\mathrm{stab}(x) & 0 & 0 \\\hline
            0 & \frac{1}{\tau(x)} \Pi_r & 0 \\\hline
            0 & 0 & \frac{\rho}{u_\dd(x)} P^{-2}
        \end{array}\right]
        \tilde{\Psi}
        \preceq 0
    \end{equation}
    with 
    \small
    \begin{equation*}
        \tilde{\Psi}^\top
        =
        \left[\begin{array}{cc|ccc|c}
            I_n & \Xi(x)^\top & u_\dd(x) I_n & K_\mathrm{n}(x)^\top & 0 & u_\dd(x) I_n \\
            0 & I_n & 0 & 0 & I_n & 0
        \end{array}\right],
    \end{equation*}
    \normalsize
    where $\tilde{\Psi}$ results following the discussion in~\cite[Sec.~8.1.2]{scherer:weiland:2000}. 
    Recall that $u_\dd(x)$ is strictly SOS such that we can divide Inequality~\eqref{eq:proof-primal-inequality} by $u_\dd(x)$ leading to
    \begin{equation}\label{eq:proof-primal-inequality-scaled}
        \tilde{\Psi}^\top
        \left[\def\arraystretch{1.2}\begin{array}{c|c|c}
            \Pi_\mathrm{stab}(x) & 0 & 0 \\\hline
            0 & \frac{1}{\tau(x)u_\dd(x)} \Pi_r & 0 \\\hline
            0 & 0 & \frac{\rho}{u_\dd(x)^2} P^{-2}
        \end{array}\right]
        \tilde{\Psi}
        \preceq 0.
    \end{equation}
    Multiplying~\eqref{eq:proof-primal-inequality-scaled} from the left and from the right by $\begin{bmatrix}
        x^\top & u_\dd(x) r(x,u)^\top
    \end{bmatrix}$ and its transpose, respectively,
    and exploiting the generalized S-procedure \cite[Lm.~2.1]{tan:2006} with the rational multiplier $\frac{u_\dd(x)}{\tau(x)} > 0$ for all $x\in\bbR^n$ yields
    \begin{equation}\label{eq:Lyapunov-inequality-S-procedure}
        \left[\star\right]^\top
        \Pi_\mathrm{stab}(x)
        \begin{bmatrix}
            x \\ \Xi(x) x + u_\dd(x)r(x,u)
        \end{bmatrix}
        \leq -\rho \|P\|_2^2 \|x\|^2
    \end{equation}
    for all $x\in\bbR^n$ satisfying 
    \begin{equation}\label{eq:Lyapunov-inequality-S-procedure-uncertainty}
        \left[\star\right]^\top
        \frac{1}{u_\dd(x)^2}\Pi_r
        \begin{bmatrix}
            u_\dd(x) x \\
            K_\mathrm{n}(x) x \\
            u_\dd(x) r(x,u)
        \end{bmatrix}
        \geq 0.
    \end{equation}
    Then, we rewrite Inequality~\eqref{eq:Lyapunov-inequality-S-procedure} as
    \begin{multline}\label{eq:proof-Lyapunov-inequality}
        \left[\star\right]^\top
        \begin{bmatrix}
            -P^{-1} & 0 \\ 
            0 & P^{-1}
        \end{bmatrix}
        \begin{bmatrix}
            x \\ \frac{1}{u_\dd(x)}\Xi(x) x + r(x,u)
        \end{bmatrix}
        \\
        \leq -\rho \|P\|_2^2 \|x\|^2
    \end{multline}
    and note that after substituting the control law
    \begin{equation*}
        \mu(x) 
        = \frac{1}{u_\dd(x)} L_\mathrm{n}(x) P^{-1} x
        = \frac{1}{u_\dd(x)} K_\mathrm{n}(x)x
    \end{equation*}
    into the dynamics~\eqref{eq:dynamics-bilinear}, $x_+$ reads
    \begin{align*}
        x_+ 
        &= Ax + \tfrac{1}{u_\dd(x)}\left( 
            B_0 K_\mathrm{n}(x)x + \tB (K_\mathrm{n}(x)x\kron x)
        \right)
        + r(x,u)
        \\
        &= \tfrac{1}{u_\dd(x)}\Xi(x) x + r(x,u).
    \end{align*}
    Thus,~\eqref{eq:proof-Lyapunov-inequality} establishes the desired Lyapunov inequality $\Delta V(x) \leq -\varepsilon\|x\|^2$ with $\varepsilon=\rho \|P\|_2^2$ for all $x\in\bbR^n$ satisfying~\eqref{eq:Lyapunov-inequality-S-procedure-uncertainty}.
    Finally, we exploit the bound in~\eqref{eq:proportional-bound-bilinear}, i.e.,
    \begin{align*}
        \|r(x,u)\|^2
        &\leq 2c_x^2 x^\top x + 2c_u^2 u^\top u
        \\
        &\leq x^\top \left(2c_x^2 I_n + \tfrac{2c_u^2}{u_\dd(x)^2} K_\mathrm{n}(x)^\top K_\mathrm{n}(x)\right) x
    \end{align*}
    corresponding to~\eqref{eq:Lyapunov-inequality-S-procedure-uncertainty} and, thus, completing the proof.
\end{proof}
We note that~\eqref{eq:controller-design-bilinear} is an SOS program for a fixed polynomial $u_\dd$, i.e., it is linear in the decision variables $P$, $L_\mathrm{n}$, $\tau$, $\rho$.
Thus, the SOS-based controller design can be solved using convex optimization techniques.
Here, the a priori chosen $u_\dd$ acts as a tuning parameter offering an additional degree of freedom in the controller parametrization and it can also be optimized when fixing the other decision variables $P$ and $\rho$.
A higher degree $\alpha$ in the controller design allows for a more flexible controller design, but it also increases the computational complexity of the SOS program.
In particular, if the controller design in Theorem~\ref{thm:controller-design-bilinear} is infeasible for a given $u_\dd$ with degree $\alpha$, one can increase $\alpha$ and re-solve the SOS program.
Note that if exponential stability is not desired but asymptotic stability is sufficient, this is achieved by setting $\rho=0$ in~\eqref{eq:controller-design-bilinear}.
The key technical novelty compared to the design schemes in~\cite{nobuyama:aoyagi:kami:2011,vatani:hovd:olaru:2014} is the consideration of \emph{uncertain} bilinear systems via SOS optimization.
Moreover, existing robust controller design methods for bilinear systems~\cite{strasser:berberich:allgower:2023b,strasser:schaller:worthmann:berberich:allgower:2024b} only admit local stability guarantees and typically a very small guaranteed RoA. 
On the contrary, the proposed approach guarantees global exponential stability of the uncertain system~\eqref{eq:dynamics-bilinear}.
As we will see in Section~\ref{sec:numerical-example-bilinear} with a numerical example, our SOS-based controller design yields, apart from the \emph{global} RoA, a clear improvement in feasibility, i.e., the design can be solved for significantly larger constants $c_x$, $c_u$ characterizing the proportional error bound~\eqref{eq:proportional-bound-bilinear}.
Additional advantages of the controller design in Theorem~\ref{thm:controller-design-bilinear} are that it can be applied to general bilinear systems~\eqref{eq:dynamics-bilinear} and straightforwardly extended to account for global performance specifications, whereas the design in~\cite{strasser:berberich:allgower:2023b} is inherently constrained to \emph{local} quadratic performance and only feasible if $(A,B_0)$ is stabilizable.
%
\section{Koopman-based control via SOS optimization}\label{sec:controller-design-nonlinear}
After establishing a guaranteed stabilizing controller design for bilinear systems, we now extend the results of Section~\ref{sec:controller-design-bilinear} to general nonlinear systems~\eqref{eq:dynamics-nonlinear} using Koopman operator theory. 
To this end, we rely on the data-driven bilinear surrogate model~\eqref{eq:dynamics-nonlinear-bilinear} obtained by SafEDMD. 
Then, we design a state-feedback controller $\mu(x)$ analogous to the bilinear case in Theorem~\ref{thm:controller-design-bilinear} to stabilize the origin of the nonlinear system~\eqref{eq:dynamics-nonlinear}. 
Here, $\mu(x)$ is rational in the lifted state $\Phi(x)$.
Before stating the main result, we introduce an assumption on the data-driven bilinear surrogate model.
\begin{assumption}\label{ass:proportional-bound}
    The proportional bound~\eqref{eq:proportional-bound-nonlinear} holds for some $c_x,c_u>0$.
\end{assumption}
Suppose the bilinear surrogate model~\eqref{eq:dynamics-nonlinear-bilinear} is obtained via SafEDMD with some $\Delta t>0$ and data $\cD$ in~\eqref{eq:dataset} sampled i.i.d. from a compact region $\bbX$ for some data length $d>0$, and the projection error onto the finite-dimensional dictionary spanned by the observable function $\Phi$ is proportionally bounded with constants $\tilde{c}_x$, $\tilde{c}_u$, e.g., trivially due to the invariance of the dictionary.
Then, Assumption~\ref{ass:proportional-bound} is satisfied for $c_x=\tilde{c}_x+\bar{c}_x$ and $c_u=\tilde{c}_u + \bar{c}_u$ with probability $1-\delta\in(0,1)$ due to~\cite[Cor.~3.2]{strasser:schaller:worthmann:berberich:allgower:2024b}, where $\bar{c}_x,\bar{c}_u\in\cO(\nicefrac{1}{\sqrt{\delta d}} + \Delta t^2)$.
For a more detailed discussion on how to establish a proportional bound on the projection error, we refer to~\cite{strasser:schaller:berberich:worthmann:allgower:2025}.
The following corollary establishes our proposed data-driven controller design for nonlinear systems. 
\begin{corollary}\label{cor:controller-design-nonlinear}
    Suppose Assumptions~\ref{ass:continuity-dynamics} and~\ref{ass:proportional-bound} hold.
    If there exist $\alpha>0$, $P=P^\top\succ 0$ of size $N\times N$, $L_\mathrm{n}\in\bbR[z,2\alpha-1]^{m\times N}$, $\tau\in\SOS_+[z,2\alpha]$, $u_\dd\in\SOS_+[z,2\alpha]$, and $\rho>0$ such that~\eqref{eq:controller-design-nonlinear}
    \begin{figure*}[tb]
        \begin{equation}\label{eq:controller-design-nonlinear}
            \begin{bmatrix}
                u_\dd(z)P - \tau(z)I_N & 0 & 0 & u_\dd(z)AP + B_0 L_\mathrm{n}(z) + \tB (L_\mathrm{n}(z)\kron z) \\
                \star & \frac{\tau(z)}{2c_x^{2}}I_N & 0 & u_\dd(z)P \\
                \star & \star & \frac{\tau(z)}{2c_u^{2}} I_m & L_\mathrm{n}(z) \\
                \star & \star & \star & u_\dd(z) (P - \rho I_N)
            \end{bmatrix}
            \in\SOS[z,2\alpha]^{3N+m}
        \end{equation}
        \medskip
        \hrule
    \end{figure*}
    holds with $z\in\bbR^N$,
    then the sampled controller
    \begin{equation}\label{eq:controller-nonlinear-explicit}
        \mu_\mathrm{s}(x(t))
        = \mu(x(k\Delta t))
        ,\quad
        t\in[k\Delta t, (k+1)\Delta t), k\geq 0
    \end{equation}
    with $\mu(x) = \frac{1}{u_\dd(\Phi(x))}L_\mathrm{n}(\Phi(x))P^{-1} \Phi(x)$ globally exponentially stabilizes the origin of the continuous-time nonlinear system~\eqref{eq:dynamics-nonlinear}.
\end{corollary}
\begin{proof}
    See~\cite{strasser:berberich:allgower:2023b,strasser:schaller:worthmann:berberich:allgower:2024b} for the necessary modifications in the proof of Theorem~\ref{thm:controller-design-bilinear}, and we refer to~\cite[Cor.~4.2]{strasser:schaller:worthmann:berberich:allgower:2024b} to derive the closed-loop stability of the continuous-time system based on sampled feedback.
\end{proof}

Corollary~\ref{cor:controller-design-nonlinear} provides a guaranteed stabilizing controller for the unknown nonlinear system~\eqref{eq:dynamics-nonlinear} based on the bilinear surrogate model~\eqref{eq:dynamics-nonlinear-bilinear} obtained by SafEDMD.
More precisely, the sampled nonlinear feedback controller~\eqref{eq:controller-nonlinear-explicit} is determined by solving the SOS program~\eqref{eq:controller-design-nonlinear} for a fixed $u_\dd$, which serves as an additional degree of freedom chosen by the user.
We emphasize that, in general, SafEDMD can only guarantee Assumption~\ref{ass:proportional-bound} with probability $1-\delta$ and on a compact region $\bbX$ from which the data are sampled.
Hence, the stability guarantees for the underlying nonlinear system~\eqref{eq:dynamics-nonlinear} hold on a suitable subset of $\bbX$ with probability $1-\delta$. 
In particular, for any solution of the SOS program in Corollary~\ref{cor:controller-design-nonlinear} there exists a constant $c>0$ such that the corresponding Lyapunov sublevel set $\Omega\coloneqq \{x\in\bbR^n\mid \Phi(x)^\top P^{-1} \Phi(x) \leq c\}\subseteq\bbX$ is invariant and, therefore, guarantees Assumption~\ref{ass:proportional-bound}. 
Thus, closed-loop stability can be guaranteed with probability $1-\delta$ for any initial condition in $\Omega$, whose size is mainly influenced by the shape of the Lyapunov function and the sampling region $\bbX$.
Analogous to~\cite[Cor.~4.2]{strasser:schaller:worthmann:berberich:allgower:2024b}, Corollary~\ref{cor:controller-design-nonlinear} establishes closed-loop stability guarantees for continuous-time systems from discrete samples, whereas most existing approaches require derivative data, which are typically hard to obtain.
If the conditions in Corollary~\ref{cor:controller-design-nonlinear} are not satisfied, a possible strategy is to repeat the controller design for a more accurat data-driven representation, i.e., by using a different lifting function $\Phi$ or a larger dataset $\cD$.
An interesting topic for future research is to analyze the feasibility of the SOS program~\eqref{eq:controller-design-nonlinear} depending on the structural properties of the underlying nonlinear system~\eqref{eq:dynamics-nonlinear}.

We note that~\cite[Cor.~4.2]{strasser:schaller:worthmann:berberich:allgower:2024b} establishes a related controller design based on linear matrix inequalities (LMIs) guaranteeing stability of the continuous-time system~\eqref{eq:dynamics-nonlinear} with probability $1-\delta$ locally around the origin.
However, the therein derived controller conservatively over-approximates the bilinearity of the surrogate model~\eqref{eq:dynamics-nonlinear-bilinear} and, thus, the guaranteed RoA may become small in order to ensure feasibility of the linear controller design. 
As key advantages over~\cite{strasser:schaller:worthmann:berberich:allgower:2024b}, the proposed controller design in Corollary~\ref{cor:controller-design-nonlinear} explicitly exploits the bilinear structure of the lifted system~\eqref{eq:dynamics-nonlinear-bilinear} without over-approximating the bilinearity, which leads to a less conservative controller design with increased feasibility.
This allows for a significantly larger guaranteed RoA and a more efficient usage of data.
In particular, by tolerating larger constants $c_x$, $c_u$ in the proportional error bound~\eqref{eq:proportional-bound-nonlinear}, fewer data samples are sufficient for a feasible controller design compared to~\cite{strasser:schaller:worthmann:berberich:allgower:2024b}.
Further, our results provide the basis for incorporating performance specifications for the underlying nonlinear system~\cite{strasser:berberich:schaller:worthmann:allgower:2025}.

We provide a summary of the proposed controller design in Algorithm~\ref{alg:SOS-controller}.
\begin{algorithm}[t]  
    \caption{SOS-based nonlinear data-driven controller design with closed-loop guarantees.}
    \begin{algorithmic}[1]
        \State Choose lifting function $\Phi:\bbR^n\to\bbR^N$, $\alpha>0$, and denominator $u_\dd\in\SOS_+[z,2\alpha]$ for $z\in\bbR^N$.
        \State Generate data $\cD$ from~\eqref{eq:dynamics-nonlinear} with sampling time $\Delta t>0$.
        \State Determine $c_x$, $c_u$ such that Assumption~\ref{ass:proportional-bound} holds.
        \If{SOS condition~\eqref{eq:controller-design-nonlinear} is feasible}\label{alg:line:controller-design}
            \State Apply controller~\eqref{eq:controller-nonlinear-explicit} to nonlinear system~\eqref{eq:dynamics-nonlinear}.
        \Else
            \State Repeat from line~1 with different parameters.
        \EndIf
    \end{algorithmic}
    \label{alg:SOS-controller}
\end{algorithm}
%
\section{Numerical examples}\label{sec:numerical-example}
Next, we illustrate the proposed controller design with numerical examples and showcase the benefits compared to existing approaches.
First, we apply the methods from Section~\ref{sec:numerical-example-bilinear} to a bilinear system.
Then, Section~\ref{sec:numerical-example-nonlinear} applies the results from Section~\ref{sec:controller-design-nonlinear} to control an unknown nonlinear system.
The simulations are conducted in MATLAB using the YALMIP toolbox~\cite{lofberg:2004} with its SOS module~\cite{lofberg:2009} and the MOSEK solver~\cite{mosek:2022} for the SOS programs.

\subsection{Bilinear system: Building control}\label{sec:numerical-example-bilinear}
We start with a discrete-time bilinear system and compare our proposed SOS-based controller design to the existing approach proposed in~\cite{strasser:berberich:allgower:2023b}. 
To this end, we consider a zone temperature process used for building control~\cite{huang:2011,yuan:perez:2006}, i.e.,
\begin{equation}\label{eq:exmp:building-control}
    x_+ 
    = x 
    + T_s V_z^{-1} u (T_0 - x)
    + r(x,u)
\end{equation}
with zone temperature $x\in\bbR$, air volume flow rate $u\in\bbR$, zone volume $V_z$, and supply air temperature $T_0$, where all variables denote deviations from the desired setpoints.
Note that $x$ should be larger than $T_0$ to ensure cooling. 
For the simulation, we choose the parameters as $V_z=2$, $T_0=-1$, $T_s=1$, and assume the system dynamics to be known except for the residual term $r$. 
For the residual, we assume $r$ to be bounded by the proportional bound~\eqref{eq:proportional-bound-bilinear} with $c_x,c_u>0$.
The goal is to design a controller that stabilizes the perturbed bilinear system~\eqref{eq:exmp:building-control} for all possible perturbations $r$ satisfying the bound~\eqref{eq:proportional-bound-bilinear}.
To this end, we choose the denominator $u_\dd\in\SOS_+[x,2\alpha]$ in the controller~\eqref{eq:controller-bilinear-explicit} as $u_\dd(x)=0.01 + (1+x)^{2\alpha}$ for some $\alpha\in\bbN$.
Then, we report the feasibility of the SOS program~\eqref{eq:controller-design-bilinear} for chosen denominator degrees $2\alpha\in\{2,4,6,8\}$ and varying constants $c_x$, $c_u$ in the proportional bound~\eqref{eq:proportional-bound-bilinear}.
While our proposed controller design is guaranteed to globally stabilize the system, the approach in~\cite{strasser:berberich:allgower:2023b} is only locally stabilizing the system in a pre-defined region.
\begin{table}[tb]
    \centering
    \caption{Computation times of the designs in Theorem~\ref{thm:controller-design-bilinear} and~\cite{strasser:berberich:allgower:2023b}.}
    \label{tab:numerics-bilinear-computation-time}
    \begin{tabular}{c|ccccc}
        & $\alpha=1$ & $\alpha=2$ & $\alpha=3$ & $\alpha=4$ & \cite{strasser:berberich:allgower:2023b} \\\hline 
        time [s] & 0.3624 & 0.3848 & 0.4180 & 0.4765 & 0.1225
    \end{tabular}
\end{table}
\begin{figure}[tb]
    \centering
    \input{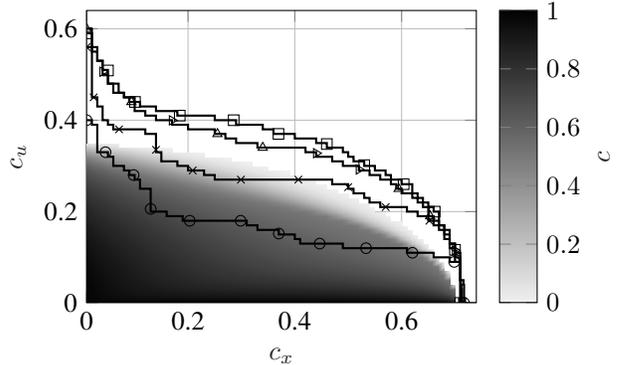}
    \caption{Feasibility of the proposed SOS-based controller design with $u_\dd\in\SOS_+(x,2\alpha)$ and the approach in~\cite{strasser:berberich:allgower:2023b}.
    Here, every combination of $c_x$, $c_u$ below the solid lines leads to a feasible and globally exponentially stabilizing controller for $\alpha=1$~\eqref{plot:bilinear:SOS:deg2}, $\alpha=2$~\eqref{plot:bilinear:SOS:deg4}, $\alpha=3$~\eqref{plot:bilinear:SOS:deg6}, $\alpha=4$~\eqref{plot:bilinear:SOS:deg8}, whereas the colored region~(\scalebox{0.5}{\ref{plot:bilinear:LMI}}) marks the feasibility of~\cite{strasser:berberich:allgower:2023b} with local stability guarantees within a region $\|x\|^2\leq c$, where the brightness indicates the value of $c$.}
    \label{fig:numerics-bilinear}
\end{figure}
For the comparison, we optimize the RoA of the controller in~\cite{strasser:berberich:allgower:2023b} to be as large as possible, where we use the design parameters $Q_z=-I$, $S_z=0$, $R_z=c$ and maximize the volume of the ellipsoidal RoA $\|x\|^2\leq c$ for each combination of $c_x$, $c_u$.
Fig.~\ref{fig:numerics-bilinear} shows the feasibility of the proposed SOS-based controller design and the approach in~\cite{strasser:berberich:allgower:2023b}. 
We emphasize that the latter has only local stability guarantees, whereas the proposed SOS-based controller design ensures \emph{global} stability of the system for all perturbations $r$ satisfying the bound~\eqref{eq:proportional-bound-bilinear} with even larger $c_x$, $c_u$ and, therefore, outperforms the existing approach. 
The computation times for different choices of $\alpha$ are summarized in Table~\ref{tab:numerics-bilinear-computation-time} exemplarily for $c_x=c_u=0.1$.
Note that a larger denominator degree $2\alpha$ yields an increased feasibility for larger constants $c_x$, $c_u$ but results in larger computation times, showing the trade-off between conservatism and the computational complexity.

\subsection{Nonlinear system: Inverted pendulum}\label{sec:numerical-example-nonlinear}
Next, we apply Corollary~\ref{cor:controller-design-nonlinear} to control an unknown nonlinear system. 
We consider the well-studied nonlinear inverted pendulum (see~\cite{martin:schon:allgower:2023b} and the reference therein) with the continuous-time dynamics
\begin{align*}
    \dot{x}_1 
    &= x_2,
    \\
    \dot{x}_2 
    &= \frac{g}{l}\sin(x_1) - \frac{b}{ml^2}x_2 + \frac{1}{ml^2} u.
\end{align*}
For the simulation, we choose $m=1$, $l=1$, $b=0.5$, $g=9.81$, but assume the system dynamics to be unknown except for a dataset of input-state measurements as in~\eqref{eq:dataset}. 
To this end, for each constant input $u\equiv \bar{u}\in\{0,1\}$ we collect $d=200$ data pairs $(x,x_+)$ with $x$ i.i.d. sampled from $\bbX=[-\pi,\pi]^2$.  
Then, we define the lifting function $
    \Phi(x) = \begin{bmatrix} x_1 & x_2 & \sin(x_1) \end{bmatrix}^\top
$
and apply SafEDMD to generate a discrete-time bilinear surrogate model by solving the linear regression problem~\eqref{eq:EDMD-linear-regression}.
\begin{figure}[tb]
    \centering
    \input{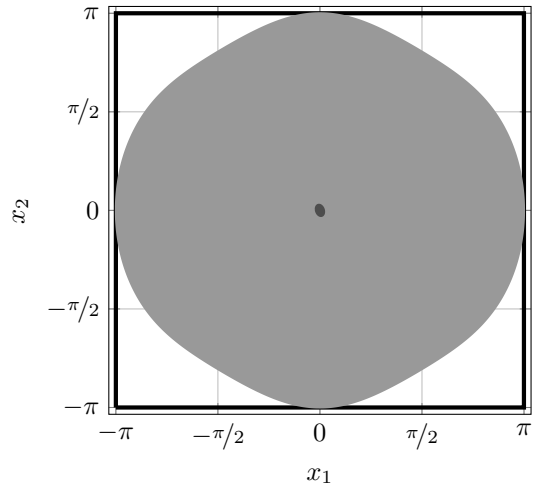}
    \caption{Compact sampling region $\bbX$~\eqref{plot:nonlinear-compactX} and guaranteed RoAs of the SOS-based~\eqref{plot:nonlinear-SOS} and LMI-based controllers~\eqref{plot:nonlinear-LMI}.}
    \label{fig:numerics-nonlinear}
\end{figure}
For the residual $r$, we consider Assumption~\ref{ass:proportional-bound} with the constants $c_x=\SI{1e-2}{}$ and $c_u=\SI{1e-3}{}$ for the proportional bound~\eqref{eq:proportional-bound-nonlinear}.
Then, we follow Corollary~\ref{cor:controller-design-nonlinear} to design an SOS-based rational controller in the lifted state $\Phi(x)$ stabilizing the underlying nonlinear system.
For the controller design, we choose the denominator by including all monomials in $\bbR[\Phi(x),2]$, i.e.,
\begin{equation*}
    u_\dd
    = 1 + x_1^2 + x_1x_2 + x_2^2 + x_1\sin(x_1) + x_2\sin(x_1) + \sin(x_1)^2,
\end{equation*}
where $u_\dd$ can be easily verified to be strictly SOS, i.e., $u_\dd\in\SOS_+[\Phi(x),2]$.
Fig.~\ref{fig:numerics-nonlinear} depicts the guaranteed RoA of the proposed SOS-based controller design for the unknown nonlinear system, where we maximized the volume of the Lyapunov sublevel set $\Omega$ in the SOS program. 

Further, we emphasize that most of the existing methods lack closed-loop guarantees for the underlying nonlinear system.
We compare our method to~\cite{strasser:schaller:worthmann:berberich:allgower:2024b} which in turn was compared extensively to Koopman-based control methods that do not provide closed-loop guarantees (like EDMDc~\cite{brunton:brunton:proctor:kutz:2016}) and therefore fail.
In particular, we compare the resulting guaranteed RoA to the RoA of the Koopman-based controller design~\cite[Cor.~4.2]{strasser:schaller:worthmann:berberich:allgower:2024b}, which relies on solving LMIs by over-approximating the bilinearity of the surrogate model.
We note that the LMI-based controller design is not feasible for the chosen constants $c_x$, $c_u$ with the data length $d$ and, thus, cannot guarantee closed-loop stability of the nonlinear system based on the collected dataset. 
Instead, we solve the controller design of~\cite{strasser:schaller:worthmann:berberich:allgower:2024b} with the smaller constants $c_x=\SI{2e-3}{}$, $c_u=\SI{2e-4}{}$, and over-approximate the bilinearity within the ellipsoidal region $\|\Phi(x)\|^2\leq 0.01$.
As seen in Fig.~\ref{fig:numerics-nonlinear}, the LMI-based controller design yields a significantly smaller RoA compared to the proposed SOS-based controller design. 
Thus, the proposed SOS-based controller design outperforms the existing approach by providing a guaranteed stabilizing controller for the unknown nonlinear system within a significantly larger RoA and based on a smaller dataset, showing the improved robustness w.r.t. the uncertainty in the bilinear surrogate model.
%
\section{Conclusion}\label{sec:conclusion}
In this paper, we presented a novel approach for designing controllers for unknown nonlinear systems using the Koopman operator.
In particular, we first used SafEDMD to derive a data-driven bilinear surrogate model with error certificates.
Then, we developed a novel SOS approach for designing rational controllers with stability guarantees for the perturbed bilinear surrogate model and, thus, the true nonlinear system.
The derived SOS condition is solvable via semidefinite programming, leading to a convex controller design for unknown nonlinear systems based on finite data.
Compared to existing approaches which rely on over-approximating the bilinearity, the proposed method reduces conservatism and enlarges the RoA by explicitly accounting for the bilinearity in the controller design.
Further, numerical examples demonstrated the reduced conservatism and improved data efficiency of the method.

\bibliographystyle{IEEEtran}
\bibliography{literature}

\begin{thebibliography}{10}
\providecommand{\url}[1]{#1}
\csname url@samestyle\endcsname
\providecommand{\newblock}{\relax}
\providecommand{\bibinfo}[2]{#2}
\providecommand{\BIBentrySTDinterwordspacing}{\spaceskip=0pt\relax}
\providecommand{\BIBentryALTinterwordstretchfactor}{4}
\providecommand{\BIBentryALTinterwordspacing}{\spaceskip=\fontdimen2\font plus
\BIBentryALTinterwordstretchfactor\fontdimen3\font minus
  \fontdimen4\font\relax}
\providecommand{\BIBforeignlanguage}[2]{{%
\expandafter\ifx\csname l@#1\endcsname\relax
\typeout{** WARNING: IEEEtran.bst: No hyphenation pattern has been}%
\typeout{** loaded for the language `#1'. Using the pattern for}%
\typeout{** the default language instead.}%
\else
\language=\csname l@#1\endcsname
\fi
#2}}
\providecommand{\BIBdecl}{\relax}
\BIBdecl

\bibitem{khalil:2002}
H.~K. Khalil, \emph{Nonlinear systems}, 3rd~ed.\hskip 1em plus 0.5em minus
  0.4em\relax Upper Saddle River, NJ: Prentice-Hall, 2002.

\bibitem{martin:schon:allgower:2023b}
T.~Martin, T.~B. Sch{\"o}n, and F.~Allg{\"o}wer, ``Guarantees for data-driven
  control of nonlinear systems using semidefinite programming: A survey,''
  \emph{Annual Reviews in Control}, p. 100911, 2023.

\bibitem{koopman:1931}
B.~O. Koopman, ``Hamiltonian systems and transformation in {Hilbert} space,''
  \emph{Proc. Nat. Acad. of Sci. USA}, vol.~17, no.~5, p. 315, 1931.

\bibitem{mauroy:mezic:susuki:2020}
A.~Mauroy, I.~Mezi{\'c}, and Y.~Susuki, \emph{The {Koopman} operator in systems
  and control}.\hskip 1em plus 0.5em minus 0.4em\relax Springer, 2020.

\bibitem{budisic:mohr:mezic:2012}
M.~Budi{\v{s}}i{\'c}, R.~Mohr, and I.~Mezi{\'c}, ``Applied {Koopman}ism,''
  \emph{Chaos: An Interdisciplinary Journal of Nonlinear Science}, vol.~22,
  no.~4, 2012.

\bibitem{mezic:2013}
I.~Mezi{\'c}, ``Analysis of fluid flows via spectral properties of the
  {Koopman} operator,'' \emph{Annual review of fluid mechanics}, vol.~45, pp.
  357--378, 2013.

\bibitem{bruder:remy:vasudevan:2019}
D.~Bruder, C.~D. Remy, and R.~Vasudevan, ``Nonlinear system identification of
  soft robot dynamics using {Koopman} operator theory,'' in \emph{Proc. IEEE
  Int. Conf. Robot. Automat. (ICRA)}, 2019, pp. 6244--6250.

\bibitem{golany:radinstky:freedman:minha:2021}
T.~Golany, K.~Radinsky, D.~Freedman, and S.~Minha, ``12-lead ecg reconstruction
  via {Koopman} operators,'' in \emph{International Conference on Machine
  Learning}.\hskip 1em plus 0.5em minus 0.4em\relax PMLR, 2021, pp. 3745--3754.

\bibitem{strasser:schlor:allgower:2023}
R.~Str{\"a}sser, S.~Schlor, and F.~Allg{\"o}wer, ``Decrypting nonlinearity:
  {Koopman} interpretation and analysis of cryptosystems,'' \emph{Automatica},
  vol. 173, p. 112022, 2025.

\bibitem{williams:kevrekidis:rowley:2015}
M.~Williams, I.~Kevrekidis, and C.~Rowley, ``A data-driven approximation of the
  {Koopman} operator: Extending dynamic mode decomposition,'' \emph{J.
  Nonlinear Science}, vol.~25, no.~6, pp. 1307--1346, 08 2015.

\bibitem{proctor:brunton:kutz:2018}
J.~L. Proctor, S.~L. Brunton, and J.~N. Kutz, ``Generalizing {Koopman} theory
  to allow for inputs and control,'' \emph{SIAM Journal on Applied Dynamical
  Systems}, vol.~17, no.~1, pp. 909--930, 2018.

\bibitem{bevanda:sosnowski:hirche:2021}
P.~Bevanda, S.~Sosnowski, and S.~Hirche, ``{Koopman} operator dynamical models:
  Learning, analysis and control,'' \emph{Annual Reviews in Control}, vol.~52,
  pp. 197--212, 2021.

\bibitem{eyuboglu:karimi:2024}
M.~Ey{\"u}bo{\u g}lu and A.~Karimi, ``Koopman-based data-driven robust control
  of nonlinear systems using integral quadratic constraints,'' 2024.

\bibitem{eyuboglu:strasser:allgower:karimi:2025}
M.~Eyübo\u{g}lu, R.~Str{\"a}sser, F.~Allg{\"o}wer, and A.~Karimi,
  ``Koopman-based {LPV} control: A data-driven approach using {IQCs},'' 2025.

\bibitem{strasser:berberich:allgower:2023a}
R.~Str{\"a}sser, J.~Berberich, and F.~Allg{\"o}wer, ``Robust data-driven
  control for nonlinear systems using the {Koopman} operator,''
  \emph{IFAC-PapersOnLine}, vol.~56, no.~2, pp. 2257--2262, 2023.

\bibitem{strasser:schaller:worthmann:berberich:allgower:2025}
R.~Str{\"a}sser, M.~Schaller, K.~Worthmann, J.~Berberich, and F.~Allgöwer,
  ``{Koopman}-based feedback design with stability guarantees,'' \emph{IEEE
  Transactions on Automatic Control}, vol.~70, no.~1, pp. 355--370, 2025.

\bibitem{strasser:schaller:worthmann:berberich:allgower:2024b}
------, ``{SafEDMD}: A {Koopman}-based data-driven controller design framework
  for nonlinear dynamical systems,'' \emph{arXiv:2402.03145}, 2024.

\bibitem{chatzikiriakos:strasser:iannelli:allgower:2024}
N.~Chatzikiriakos, R.~Str{\"a}sser, F.~Allg{\"o}wer, and A.~Iannelli,
  ``End-to-end guarantees for indirect data-driven control of bilinear systems
  with finite stochastic data,'' \emph{arXiv:2409.18010}, 2024.

\bibitem{worthmann:strasser:schaller:berberich:allgower:2024}
K.~Worthmann, R.~Strässer, M.~Schaller, J.~Berberich, and F.~Allgöwer,
  ``Data-driven {MPC} with terminal conditions in the {Koopman} framework,'' in
  \emph{Proc. 63rd IEEE Conference on Decision and Control (CDC)}, 2024, pp.
  146--151.

\bibitem{nobuyama:aoyagi:kami:2011}
E.~Nobuyama, T.~Aoyagi, and Y.~Kami, ``A sum of squares optimization approach
  to robust control of bilinear systems,'' in \emph{Recent Advances in Robust
  Control-Theory and Applications in Robotics and Electromechanics}.\hskip 1em
  plus 0.5em minus 0.4em\relax IntechOpen, 2011.

\bibitem{vatani:hovd:olaru:2014}
M.~Vatani, M.~Hovd, and S.~Olaru, ``Control design and analysis for discrete
  time bilinear systems using sum of squares methods,'' in \emph{Proc. 53rd
  IEEE Conf. Decision and Control (CDC)}, 2014, pp. 3143--3148.

\bibitem{boyd:vandenberghe:2004}
S.~P. Boyd and L.~Vandenberghe, \emph{Convex Optimization}.\hskip 1em plus
  0.5em minus 0.4em\relax Cambridge University Press, 2004.

\bibitem{scherer:weiland:2000}
C.~W. Scherer and S.~Weiland, ``Linear matrix inequalities in control,''
  \emph{Lecture Notes, Dutch Institute for Systems and Control, Delft, The
  Netherlands}, vol.~3, no.~2, 2000.

\bibitem{tan:2006}
W.~Tan, ``Nonlinear control analysis and synthesis using sum-of-squares
  programming,'' {Ph.D.} dissertation, University of California, Berkeley,
  2000.

\bibitem{strasser:berberich:allgower:2023b}
R.~Str{\"a}sser, J.~Berberich, and F.~Allgöwer, ``Control of bilinear systems
  using gain-scheduling: {Stability} and performance guarantees,'' in
  \emph{Proc. 62nd IEEE Conf. Decision and Control (CDC)}, 2023, pp.
  4674--4681.

\bibitem{strasser:schaller:berberich:worthmann:allgower:2025}
R.~Str{\"a}sser, M.~Schaller, J.~Berberich, K.~Worthmann, and F.~Allg{\"o}wer,
  ``Kernel-based error bounds of bilinear {Koopman} surrogate models for
  nonlinear data-driven control,'' \emph{IEEE Control Systems Letters}, vol.~9,
  pp. 1892--1897, 2025.

\bibitem{strasser:berberich:schaller:worthmann:allgower:2025}
R.~Str{\"a}sser, J.~Berberich, M.~Schaller, K.~Worthmann, and F.~Allg{\"o}wer,
  ``{Koopman}-based control of nonlinear systems with closed-loop guarantees,''
  \emph{at-Automatisierungstechnik}, vol.~73, no.~6, pp. 413--428, 2025.

\bibitem{lofberg:2004}
J.~{L\"{o}fberg}, ``{YALMIP}: A toolbox for modeling and optimization in
  {MATLAB},'' in \emph{Proc. {IEEE} International Conference on Robotics and
  Automation (ICRA)}, 2004, pp. 284--289.

\bibitem{lofberg:2009}
J.~L{\"{o}}fberg, ``Pre- and post-processing sum-of-squares programs in
  practice,'' \emph{IEEE Transactions on Automatic Control}, vol.~54, no.~5,
  pp. 1007--1011, 2009.

\bibitem{mosek:2022}
M.~ApS, \emph{The {MOSEK} optimization toolbox for {MATLAB} manual. Version
  9.3.21}, 2022.

\bibitem{huang:2011}
G.~Huang, ``Model predictive control of {VAV} zone thermal systems concerning
  bi-linearity and gain nonlinearity,'' \emph{Control engineering practice},
  vol.~19, no.~7, pp. 700--710, 2011.

\bibitem{yuan:perez:2006}
S.~Yuan and R.~Perez, ``Multiple-zone ventilation and temperature control of a
  single-duct {VAV} system using model predictive strategy,'' \emph{Energy and
  buildings}, vol.~38, no.~10, pp. 1248--1261, 2006.

\bibitem{brunton:brunton:proctor:kutz:2016}
S.~L. Brunton, B.~W. Brunton, J.~L. Proctor, and J.~N. Kutz, ``{Koopman}
  invariant subspaces and finite linear representations of nonlinear dynamical
  systems for control,'' \emph{PloS one}, vol.~11, no.~2, pp. 1--19, 2016.

\end{thebibliography}

\end{document}